\documentclass[fleqn]{eptcs}
\providecommand{\event}{GandALF 2010}
\usepackage{breakurl}
\usepackage{amsmath,amssymb,amsthm,algorithm,algorithmic}

\usepackage{tikz}
\usetikzlibrary{decorations.pathmorphing}
\usetikzlibrary{shapes.geometric}
\usetikzlibrary{arrows}
\tikzset{p0/.style={shape=circle,draw,thick,minimum size = 1.2cm}}
\tikzset{p1/.style={shape=rectangle,minimum size = 1.2cm,draw,thick}}

\renewcommand{\epsilon}{\varepsilon}
\newcommand{\F}{\ensuremath{\mathcal{F}}}

\newcommand{\pow}[1]{2^{#1}}
\newcommand{\nats}{\mathbb{N}}

\newcommand{\game}{\mathcal{G}}
\newcommand{\att}[3]{\mathrm{Attr}_{#1}^{#2}(#3)}
\newcommand{\ztree}{\mathcal{Z}_{\F_0, \F_1}}

\DeclareMathOperator{\infi}{\mathrm{Inf}}
\DeclareMathOperator{\occ}{\mathrm{Occ}}
\DeclareMathOperator{\maxscore}{\mathrm{MaxSc}}
\DeclareMathOperator{\play}{\mathrm{Play}}
\DeclareMathOperator{\lbl}{\mathrm{RtLbl}}
\DeclareMathOperator{\child}{\mathrm{Chld}}
\DeclareMathOperator{\branch}{\mathrm{BrnchFctr}}
\DeclareMathOperator{\score}{\mathrm{Sc}}
\DeclareMathOperator{\acc}{\mathrm{Acc}}
\DeclareMathOperator{\ind}{\mathrm{Ind}}

\newtheorem{theorem}{Theorem}
\newtheorem{lemma}[theorem]{Lemma}
\newtheorem{definition}[theorem]{Definition}

\title{Playing Muller Games in a Hurry
\thanks{This work was carried out while the second author visited
the University of Warwick, supported by EPSRC grant EP/E022030/1 and 
the project 
\textit{Games for Analysis and Synthesis of Interactive Computational Systems
(GASICS)} of the \textit{European Science Foundation}.}}

\author{John Fearnley
\institute{Department of Computer Science\\
 University of Warwick, UK}
\email{john@dcs.warwick.ac.uk}
\and 
Martin Zimmermann
\institute{Lehrstuhl Informatik 7\\
 RWTH Aachen University, Germany}
\email{zimmermann@automata.rwth-aachen.de}
}
\def\authorrunning{J. Fearnley \& M. Zimmermann}
\def\titlerunning{Playing Muller Games in a Hurry}

\begin{document}
\maketitle

\begin{abstract}
This work studies the following question: can plays
in a Muller game be stopped after a finite number of moves
and a winner be declared. A criterion to do this is sound
if Player $0$ wins an infinite-duration Muller game 
if and only if she wins the finite-duration version. 
A sound criterion is presented that stops a play after
at most $3^{n}$ moves, where $n$ is the size of the arena. This improves the bound $(n!+1)^{n}$
obtained by McNaughton and the bound $n!+1$ derived from a 
reduction to parity games.
\end{abstract}

\section{Introduction}

In an infinite game, two players move a token through a finite graph thereby building an infinite path. The winner is determined by a partition of the infinite paths through the arena into the paths that are winning for Player $0$ or winning for Player $1$, respectively. Many winning conditions in the literature depend on the vertices that are visited infinitely often, i.e., the
winner of a play cannot be determined after any finite number of steps. We are interested in the following question: is it nevertheless possible to give a criterion to define a finite-duration variant of an infinite game. Such a criterion has to stop a play after a finite number of steps and then declare a winner based on the finite play constructed thus far. It is sound if Player~$0$ has a winning strategy for the infinite-duration game if and only if she has a winning strategy for the finite-duration game.

McNaughton considered the problem of playing infinite games in finite time from a different perspective. His motivation was to make infinite games suitable for ``casual living room recreation'' \cite{M00}.
As human players cannot play infinitely long, he envisions a referee that stops a
play at a certain time and declares a winner. The justification for declaring
a winner is that ``if the play were to continue with each [player] playing
forever as he has so far, then the player declared to be the winner would be the winner of the infinite play of the game'' \cite{M00}. 

Besides this recreational aspect of infinite games there are several interesting theoretical questions that motivate investigating this problem. If there exists a sound criterion to stop a play after at most $n$ steps, this yields a simple algorithm to determine the winner of the infinite game: the finite-duration game can be seen as a reachability game on a finite tree of depth at most $n$ that is won by the same player that wins the infinite-duration game. There exist simple and efficient algorithms to determine the winner in reachability games on trees. Furthermore, a positive answer to the question whether a winning strategy for the reachability game can be turned into a (small finite-state) winning strategy should yield better results in the average (although not in the worst case) than game reductions, which  ignore the structure of the arena.

Consider the following criterion: the players move the token
through the arena until a vertex is visited for the second time.
An infinite play can then be obtained by assuming that the players continue to
play the cycle that they have constructed. Then, the winner of
the infinite play is declared to be the winner of the finite play. If the game is determined
with positional strategies for both players, then this procedure is correct: if a player has a
winning strategy for the infinite game, which can be assumed to be positional,
then she can use the same strategy to win the finite version of the game
and vice versa. 

Therefore, McNaughton proposes that we should consider games that are
in general not positionally determined. Here, the first loop of a 
play is typically not an indicator of how the infinite play evolves, 
as the memory allows 
a player to make different decisions when a vertex is seen again.
Therefore, the players have to play longer before the play 
can be stopped and analyzed.

McNaughton considers Muller games, which are games of the form $(G,\F_0,\F_1)$, 
where $G$ is a finite arena and $(\F_0,\F_1)$ is a partition of the set of vertices.
Player~$i$ wins a play, if the set of vertices visited infinitely often by
this play is in $\F_i$. Muller winning conditions allow us to express all other 
winning conditions that depend only in the infinity set of a play 
(e.g., B\"uchi, co-B\"uchi, parity, Rabin, and Streett conditions).

To give a sound criterion for Muller games, McNaughton defines
for every set of vertices $F$ a scoring function $\score_F$ that keeps track
of the number of times the set $F$ was visited entirely since the last visit of a vertex
that is not in $F$. In an infinite play, the set of vertices seen infinitely
often is the unique set $F$ such that $\score_F$ will tend to infinity with being reset to $0$ only finitely often.   



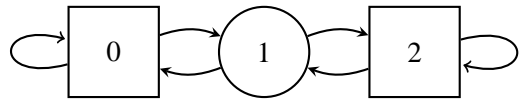
\begin{figure}[h]
\centering
\begin{tikzpicture}[node distance=2cm,auto,every edge/.style={bend angle=20,draw,thick}]
\node[p1]			(0)	{$0$};
\node[p0, right of = 0]		(1)	{$1$};
\node[p1, right of = 1]		(2)	{$2$};

\path[-stealth,thick]
(0)	edge[bend left] (1)
(1)	edge[bend left] (0)
(1)	edge[bend left] (2)
(2)	edge[bend left] (1)
(0)	edge[loop left] ()
(2)	edge[loop right] ();

\end{tikzpicture}
\caption{The arena $G$.}
\label{fig_exbound}
\end{figure}

Let $G$ be the arena in Figure~\ref{fig_exbound} (Player~$0$'s vertices are shown as 
circles and Player~$1$'s vertices are shown as squares) and the Muller game 
$\game=(G,\F_0,\F_1)$ with $\F_0= \{\{0,1,2\},\{0\},\{2\}\}$. 
In the play $100122121$ we have that the score for the set $\{1,2\}$ is $3$,
as it was seen thrice (i.e., with the infixes $12$, $21$, and $21$). Note that the order of
the visits to the elements of $F$ is irrelevant and that it is not required to close
a loop in the arena. The following winning strategy for Player~$0$ 
bounds the scores of Player~$1$ by $2$: arriving from $0$ at $1$ move to $2$ and
vice versa. However, Player~$0$ cannot avoid a score of $2$ for her opponent,
as either the play prefix $1001$ or $1221$ is consistent with every winning strategy.

By using finite-state determinacy of Muller games, McNaughton suggests that the criterion should
stop a play after a score of $|F|!+1$ for some set $F$ is reached. He shows that picking the winner
to be the Player~$i$ such that $F\in\F_i$ is indeed sound.

Applying
finite-state determinacy one can also show that one can soundly declare a
winner after at most $|G|!+1$ steps, as a repetition of a memory state has occurred
after that many steps. Note that for \textit{large} sets $F$, it could take far 
more than $|G|!+1$ steps to reach a score of $|F|!+1$, as scores can increase 
slowly or can even be reset to $0$.
However, to decide whether a memory state repetition has occurred, it might be necessary
 to compute the complete memory structure for the given game, which is of size $|G|!$.
Keeping track of scores is much simpler, as they can be computed on the fly
while the play is being played. Also, there are at most $|G|$ sets $F$ with non-zero score.

\paragraph{Our contribution.} We show that declaring the winner of a play
as soon as the score~$3$ is reached for the first time is a sound criterion. We
complement this by proving that a score of $3$ is reached
after at most $3^{|G|}$ steps. Hence, we
obtain a better bound than $|G|!+1$, which was derived from
waiting for repetitions of memory states.

Our results are obtained by
using Zielonka's algorithm \cite{Z98} (a reinterpretation of an earlier
algorithm by McNaughton \cite{M93})
for computing winning regions in Muller
games. We carefully define a winning strategy that
bounds the scores of the opponent by $2$. 
In the example above, the winning player cannot avoid a
score of $2$ for the opponent. Hence, in this sense our result is optimal.

\paragraph{Related work.} Usually, the quality of a strategy
is measured in terms of memory needed to implement it. However,
there are other natural quality measures of winning strategies.
In \cite{CH09}, the authors study a strengthening of parity (and Streett) 
objectives, which require that there is
some bound between the occurrences of even colors. Another application
of this concept appears in work on request response games~\cite{HT08,Z09}, where 
waiting times between requests and subsequent responses are used
to define the value of a play. There it is shown 
that time-optimal winning strategies can be computed effectively.

The maximal score achieved by the opponent in a play 
can be used to measure the quality of winning plays in a Muller game. 
Player~$0$ prefers plays in which the scores of her
opponent are \textit{small}. This corresponds to not spending a prolonged amount of time in a set of the opponent, but
visiting every vertex that is seen infinitely often without large gaps.\medskip

This paper is structured as follows. Section~\ref{sec_def} contains
basic definitions and fixes our notation. In Section~\ref{sec_ftmg},
we introduce the scoring functions, prove some properties about
scoring and define finite-time Muller games. In Section~\ref{sec_alg}, we
present Zielonka's algorithm which is used in Section~\ref{sec_bound2}
to prove the main result. Section~\ref{sec_conc} ends the paper with
a conclusion and some pointers to further research.

\section{Definitions}
\label{sec_def}

The power set of a set $S$ is denoted by $\pow{S}$ and $\nats$ denotes the non-negative
integers. The prefix relation on words is denoted by $\sqsubseteq$, its strict
version by $\sqsubset$. Given a word $w=xy$, define $x^{-1}w=y$ and $wy^{-1}=x$.

An arena $G=(V,V_0,V_1,E)$ consists of a finite, directed graph $(V,E)$ of vertices and
a partition $(V_0,V_1)$ of $V$ denoting the positions of Player $0$ (drawn as circles)
and Player $1$ (drawn as squares). We require that every
vertex has at least one outgoing edge. A set $X\subseteq V$ induces the subarena
$G[X]=(V\cap X, V_0\cap X,V_1\cap X, E\cap (X\times X))$, if every vertex in $X$
has at least one successor in $X$. A Muller game $\game=(G,\F_0,\F_1)$ consists of an arena $G$ and a partition $(\F_0,\F_1)$ of $\pow{V}$.

A play in $G$ starting in $v\in V$ is an infinite sequence $\rho=\rho_0\rho_1\rho_2\ldots$ such that $\rho_0=v$
and $(\rho_n,\rho_{n+1})\in E$ for all $n\in\nats$. The occurrence set $\occ(\rho)$ and infinity set
$\infi(\rho)$ of $\rho$ are given by $\occ(\rho)=\{v\in V\mid\exists n\in\nats \text{ such that }\rho_n=v\}$
and $\infi(\rho)=\{v\in V\mid\exists^{\omega} n\in\nats \text{ such that }\rho_n=v\}$.
We will also use the occurrence set of a finite play $w$. A play $\rho$ in a Muller game is winning for Player~$i$ if $\infi(\rho)\in\F_i$.

A strategy for Player $i$ is a function $\sigma\colon V^*V_i\rightarrow V$
satisfying $(s,\sigma(ws))\in E$ for all $ws\in V^*V_i$. The play $\rho$ is
consistent with $\sigma$ if $\rho_{n+1}=\sigma(\rho_0\ldots\rho_{n})$ for every
$n\in\nats$ with $\rho_n \in V_i$. The set of strategies for Player~$i$ is denoted by $\Pi_i$.
A strategy is called finite-state, if it can be implemented by an automaton with output
that reads finite plays and outputs the vertex to move to. We will say that
a finite-state strategy is of size $n$, if there exists an automaton 
with $n$ states that implements it.

A strategy $\sigma$ for Player $i$ is a winning strategy from a vertex $v\in V$, if every
play that starts in $v$ and is consistent with $\sigma$ is won by Player $i$. The strategy $\sigma$ is a
winning strategy for a set of vertices $W\subseteq V$, if every play that starts in some $v\in W$ and is
consistent with $\sigma$ is won by Player $i$. The winning region $W_i$ of Player~$i$ contains all vertices,
from which she has a winning strategy. A game is determined if $W_0$ and $W_1$ form a partition of $V$.

\begin{theorem}[\cite{BL69,GH82,M93}]
Muller games are determined with finite-state strategies of size $n\cdot n!$, where $n$ denotes the size of the arena.
\end{theorem}

Let $G=(V,V_0,V_1,E)$ be an arena and let $X\subseteq V$ be a set that induces a subarena. 
The attractor for Player $i$ of a set $F\subseteq V$ in $X$ is $\att{i}{X}{F}=\bigcup_{n=0}^{|V|}A_n$ where $A_0=F\cap X$ and 
\begin{align*}A_{n+1}=A_n\,\cup\,&\{v\in V_i\cap X\mid\exists v'\in A_n \text{ such that }(v,v')\in E\}\\
                         \cup\,&\{v\in V_{1-i}\cap X\mid\forall v'\in X \text{ with } (v,v')\in E: v'\in A_n\}\enspace.\end{align*} 
A $X\subseteq V$ is a trap for Player~$i$, if all outgoing edges of the vertices in
$V_i\cap X$ lead to $X$  and at least one successor of every vertex in
$V_{1-i}\cap X$ is in $X$.
\begin{lemma} Let $G=(V,V_0,V_1,E)$ be an arena and $F,X\subseteq V$.
\begin{enumerate}
\item For every $v\in\att{i}{X}{F}$ Player $i$ has a positional strategy to bring the play into $F$.
\item The set $V\setminus\att{i}{X}{F}$ induces a subarena and is a trap for Player~$i$ in $G$.
\end{enumerate}
\end{lemma}

\section{The Scoring Functions and Finite-time Muller Games}
\label{sec_ftmg}

This section introduces the notions that are required to formally define 
finite-time Muller games. In his study of these games, McNaughton introduced the
concept of a score. For every set of vertices $F$ we define the score of a
finite play $w$ to be the number of times that $F$ has been visited entirely
since $w$ last visited a vertex in $V \setminus F$.

\begin{definition}[Score]
For every $F\subseteq V$ we define $\score_F\colon V^+\rightarrow\nats$ as
\begin{equation*}
\score_F(w)=\max\{k\in\nats\mid \exists x_1,\ldots,x_k\in V^+ \text{ such that }
\occ(x_i)=F \text{ for all $i$ and } x_1\cdots x_k\text{ is a suffix of }
w\}.
\end{equation*}
\end{definition}

We extend this notion by introducing the concept of an accumulator. 
For every set $F$, the accumulator measures the
progress that has been made towards the next score increase of $F$.

\begin{definition}[Accumulator]
For every $F\subseteq V$ we define $\acc_F\colon V^+\rightarrow\pow{F}$ by
$\acc_F(w)=\occ(x)$, where $x$ is the longest suffix of $w$ such that
$\score_F(w)=\score_F(wy^{-1})$ for every suffix $y$ of $x$, and
$\occ(x)\subseteq F$.
\end{definition}

Finally we define the maximum score function. This function maps a subset
$\F \subseteq 2^V$ and a play $\rho$ to the highest score that is reached during $\rho$ for a set contained in $\F$.

\begin{definition}[MaxScore]
For every $\F \subseteq 2^V$ we define $\maxscore_{\F}\colon V^+\cup V^{\omega}\rightarrow\nats\cup\{\infty\}$ by $\maxscore_{\F}(\rho) = \max_{F \in \F} \max_{w \sqsubseteq \rho} \score_F(w)$.
\end{definition}

McNaughton proposes that scores should be used to decide the winner in a
finite-time Muller game. As soon as a threshold score of $k$ for some
set $F$ is reached, the play is
stopped and Player~$i$ is declared the winner, if $F\in\F_i$.
The next lemma shows that this condition is sufficient to
ensure that the game terminates after a finite number of steps.

\begin{lemma}
\label{bounded_paths}
Let $G$ be an arena with vertex set $V$. Every $w\in V^*$ with $|w|\ge k^{|V|}$ satisfies $\maxscore_{\pow{V}}(w)
\ge k$. 
\end{lemma}
\begin{proof}
We will show by induction over $|V|$ that every word
$w\in V^*$ with $|w|\ge k^{|V|}$ contains an infix $x$ that
can be decomposed as $x=x_1\cdots x_k$ where every $x_i$ is a non-empty
word with $\occ(x_i)=\occ(x)$. This will imply $\maxscore_{\pow{V}}(w)
\ge k$. 

The claim holds trivially for $|V|=1$ by choosing $x$ to be the prefix of $w$ of length $k$ and $x_i=s$ for the single vertex $s\in V$. For the induction step, consider a set $V$ with $n+1$ vertices. If $w$ contains an infix $x$ of length $k^n$ which contains at most $n$ distinct vertices, then we can apply the inductive hypothesis and obtain a decomposition of an infix of $v$ with the desired properties. Otherwise, every infix $x$ of $w$ of length $k^n$ contains every vertex of $V$ at least once. Let $x$ be the prefix of length $k^{n+1}$ of $w$ and let $x=x_1\cdots x_k$ be the decomposition of $x$ such that each $x_i$ is of length $k^n$. Then, we have $\occ(x_i)=\occ(x)=V$ for all $i$. Therefore, the decomposition has the desired properties.
\end{proof}

Lemma~\ref{bounded_paths} implies that a finite-time Muller game with threshold $k$ must
end after at most $k^{|V|}$ steps. We can also show that this bound is tight. For
every $k>0$ we give an inductive definition of a word over the alphabet
$\Sigma_n=\{1,\ldots,n\}$ by $w_{(k,1)}=1^{k-1}$ and
$w_{(k,n)}=(w_{(k,n-1)}n)^{k-1}w_{(k,n-1)}$. Clearly, the word $w_{(k, n)}$ has
length $k^n-1$, and it can also be shown that $\maxscore_{\pow{\Sigma_n}}(w)<k$.\medskip

Finally, to declare a unique winner in every finite-time Muller game we must
exclude the case where there are two sets such
that both sets hit score $k$ at the same time. McNaughton observed that for $k
\ge 2$ the first set to hit score $k$ will be unique. Before we reprove this, we
will first show a useful auxiliary result that will also be used later in the
paper.

\begin{lemma}[cf. Theorem 4.2 of \cite{M00}]\label{lem_chain}
Let $w\in V^+$.
The sets $F$ with $\score_F(w)\ge 1$ together with the sets $\acc_F(w)$ for some $F$ 
form a chain with respect to the subset relation. 
\end{lemma}
\begin{proof}
It suffices to show that all such sets are pairwise comparable: let $F$ and $F'$ be two sets such that either $\score_F(w)\ge 1$ or $F=\acc_{H}(w)$ for some $H\subseteq V$ and either $\score_{F'}(w)\ge 1$ or $F'=\acc_{H'}(w)$ for some $H'\subseteq V$. Then, there exist two decompositions $w=w_0w_1$ and $w=w_0'w_1'$ with $\occ(w_1)=F$ and $\occ(w_1')=F'$. Now, either $w_1$ is a suffix of $w_1'$ or vice versa. In the first case, we have $F\subseteq F'$ and in the second case $F'\subseteq F$.
\end{proof}

Note that Lemma~\ref{lem_chain} implies that there are at any time at most $|V|$ sets with non-zero scores.

\begin{lemma}[\cite{M00}]
\label{lem_uniqueF}
Let $k,l\ge 2$, let $F,F'\subseteq V$, let $w\in V^*$ and $v\in V$ such that $\score_F(w)<k$ and $\score_{F'}(w)<l$. If $\score_F(wv)=k$ and $\score_{F'}(wv)=l$, then $F=F'$.
\end{lemma}
\begin{proof}
Towards a contradiction assume $F\not=F'$. By Lemma~\ref{lem_chain} we can assume $F'\subset F$, i.e., there exists some $q\in F\setminus F'$. Then, $\score_F(wv)=k$ and $\score_{F'}(wv)=l$ imply the existence of decompositions $wv=w_0w_1\cdots w_k$ and $wv=w_0'w_1'\cdots w_l'$ such that $\occ(w_i)=F$ and $\occ(w_i')=F'$ for all $i\ge 1$. As $q\notin F'$, $w_1'\cdots w_l'$ is a proper suffix of $w_k$. Furthermore, as $\score_F(w)<k$, we have $v\notin\occ(w_kv^{-1})$. However, we have $v\in F'$ and hence $v\in\occ(w_{k-1}')$, which is an infix of $w_kv^{-1}$. This yields the desired contradiction.\qedhere
\end{proof}

We are now in a position to define a finite-time Muller game.
Such a game $\game=(G,\F_0,\F_1,k)$ consists of an arena
$G=(V,V_0,V_1,E)$, a partition $(\F_0,\F_1)$ of $\pow{V}$, and a threshold $k\ge
2$. By Lemma~\ref{bounded_paths} we have that every infinite play must reach
score $k$ for some set $F$ after a bounded number of steps. Therefore, we
define a play for the finite-time Muller game to be a finite path $w=w_0\cdots w_n$
with $\maxscore_{\pow{V}}(w_0\cdots w_n)=k$, but $\maxscore_{\pow{V}}(w_0\cdots w_{n-1})<k$.
Due to Lemma~\ref{lem_uniqueF}, there is a unique $F\subseteq V$ such that
$\score_F(w)=k$. Player~$0$ wins the play $w$ if $F\in\F_0$ and Player~$1$ wins
otherwise. The definitions of strategies, plays, and winning sets can
be redefined for the finite games.

Zermelo~\cite{Z13} has shown that a game in which every play is finite is
determined. Therefore, it immediately follows that finite Muller games are
determined.
\begin{lemma}
Finite-time Muller games are determined.
\end{lemma}

In fact, McNaughton considered a slightly different definition of a finite-time
Muller game. Rather than stopping the play when the score of a set reaches
the global threshold $k$, his version stops the play when the score of a set $F$
reaches $|F|! + 1$. 

\begin{theorem}[\cite{M00}]\label{thm_ftmg}
If $W_i$ is the winning region of Player~$i$ in a Muller game $(G,\F_0,\F_1)$,
and $W_i'$ is the winning region of Player~$i$ in McNaughton's finite-time
Muller game, then $W_i=W_i'$.
\end{theorem}

\section{Zielonka's Algorithm For Muller Games}
\label{sec_alg}

This section presents Zielonka's algorithm for Muller games~\cite{Z98}, 
a reinterpretation of an earlier algorithm due to McNaughton~\cite{M93}.
Our notation mostly follows \cite{DJ97,DJ98}. We will use the internal structure
of the winning regions as computed by the algorithm to define a strategy that bounds
the scores of the losing player by~$2$.

As we consider uncolored arenas, we have to deal with Muller games where
$(\F_0,\F_1)$ is a partition of $\pow{V'}$ for some finite set $V'\supseteq V$, as the 
algorithm makes recursive calls for such games. This does not change the 
semantics of Muller games, as we have $\infi(\rho)\subseteq V$ for every infinite
play $\rho$.

We begin by introducing Zielonka trees, a representation of
winning conditions $(\F_0,\F_1)$.
Given a family of sets $\F\subseteq \pow{V'}$ and $X\subseteq V'$, we define
$\F\restriction X=\{F\in\F\mid F\subseteq X\}$. Given a partition $(\F_0,\F_1)$
of $\pow{V'}$, we define $(\F_0,\F_1)\restriction X=(\F_0\restriction X,
\F_1\restriction X)$. Note that $\F\restriction X\subseteq \F$.

\begin{definition}[Zielonka tree]
For every winning condition $(\F_0, \F_1)$ defined over a set $V'$, its Zielonka tree 
$\mathcal{Z}_{\F_0, \F_1}$ is defined as follows: suppose that $V' \in \F_i$ 
and let $V_{0}', V_{1}', \dots, V_{k-1}'$ be the 
$\subseteq$-maximal sets in $\F_{1-i}$. The tree $\mathcal{Z}_{\F_0, \F_1}$ 
consists of a root vertex labelled by $V'$ with $k$ children which are 
defined by $\mathcal{Z}_{(\F_0, \F_1) \restriction V_0'}, \dots, 
\mathcal{Z}_{(\F_0, \F_1) \restriction V_{k-1}'}$.
\end{definition}
For every Zielonka tree $T$, we define $\lbl(T)$ to be the label of the root in
$T$, we define $\branch(T)$ to be the number of children that the root has in
$T$, and we define $\child(T, j)$ for $0 \le j < \branch(T)$ to be the $j$-th 
child of the root in $T$. Here, we assume that the children of every vertex are 
ordered by some fixed linear order. 

The input of Zielonka's algorithm (see Algorithm~\ref{alg_z}) is a finite arena 
$G$ with vertex set $V$ and the Zielonka tree of a partition $(\F_0,\F_1)$ of
$\pow{V'}$ for some finite set $V'\supseteq V$. The algorithm computes the winning
regions of the players by successively removing parts of Player~$0$'s winning region
(the sets $U_0,U_1,U_2,\ldots$). By doing this, the algorithm computes an internal structure
of the winning regions that will be crucial to proving our results in the next section.

For the rest of this paper
we will refer to the sets of vertices and the subtrees of $\ztree$ as computed 
by the algorithm.

\begin{algorithm}[h]
\begin{algorithmic}
\STATE $i :=$ The index $j$ such that $\lbl(\mathcal{Z}_{\F_0, \F_1}) \in \F_j$
\STATE $k := \branch(\mathcal{Z}_{\F_0 , \F_1})$
\IF{The root of $\mathcal{Z}_{\F_0, \F_1}$ has no children}
\STATE $W_i = V$; $W_{1-i} = \emptyset$
\STATE $\textbf{return} (W_0, W_1)$
\ENDIF
\STATE $U_1 := \emptyset$; $n := 0$
\REPEAT
\STATE $n := n+1$
\STATE $A_n := \att{1-i}{V}{U_{n-1}}$
\STATE $X_n := V \setminus A_{n}$
\STATE $T_n := \child(\mathcal{Z}_{\F_0 , \F_1}, n \mod k)$
\STATE $Y_n := X_n \setminus \att{i}{X_n}{V \setminus
\lbl(T_n)}$
\STATE $(W^{n}_0, W^{n}_1) := \text{Zielonka}(G[Y_n],
T_n)$
\STATE $U_n := A_{n} \cup W^{n}_{1-i}$
\UNTIL{$U_n = U_{n-1} = \dots = U_{n-k}$}
\STATE $W_i = V\setminus U_n$; $W_{1-i} = U_n$
\STATE \textbf{return} $(W_0, W_1)$
\end{algorithmic}
\caption{Zielonka$(G, \mathcal{Z}_{\F_0, \F_1})$.}
\label{alg_z}
\end{algorithm}

Figure~\ref{fig_alg} depicts the situation in the $n$-th iteration of the algorithm. The vertices
in $U_{n-1}$ have already been removed and belong to $W_{1-i}$. Then, all vertices in the $(1-i)$-attractor of $U_{n-1}$ also
belong to $W_{1-i}$. After removing these vertices from the arena, the algorithm also removes the vertices in the $i$-attractor of
$\lbl(T_n)$. The remaining vertices form a subarena whose vertex set is a subset of $\lbl(T_n)$. Hence, the algorithm can
recursively compute the winning regions $W_i^n$ in this subarena with Zielonka tree $T_n$. By construction, the winning region
$W_{1-i}^n$ is also a subset of $W_{1-i}$. This is repeated until the sets $U_n$ converge to $W_{1-i}$. All remaining vertices belong to
$W_i$. 

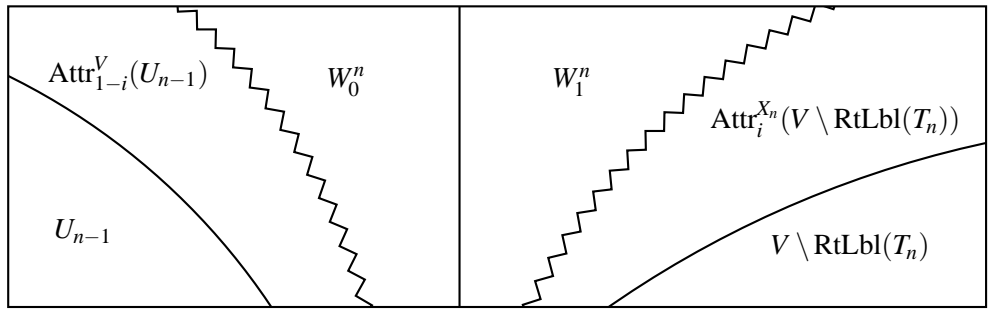
\begin{figure}[b]
\centering
\begin{tikzpicture}
\draw[thick] (0,0) rectangle (13,4);
\begin{scope}
\clip (0,0) rectangle (13,4);

\draw[thick] decorate [decoration =zigzag] {(-3,-2) circle (8cm)};
\draw[thick] (-4,-5) circle (9cm);

\node at (1,1) {$U_{n-1}$};
\node at (1.6,3.1) {$\att{1-i}{V}{U_{n-1}}$};

\draw[thick] decorate [decoration =zigzag] {(15,-4) circle (9cm)};
\draw[thick] (16,-11.5) circle (14cm);

\node at (11.2,.8) {$V\setminus\lbl(T_n)$};
\node at (11,2.5) {$\att{i}{X_n}{V\setminus\lbl(T_n)}$};

\draw[thick] (6,0) -- (6,5);

\node at (4.5,3) {$W_0^n$};
\node at (7.5,3) {$W_1^n$};

\end{scope}
\end{tikzpicture}
\caption{The sets computed by Zielonka's algorithm.}
\label{fig_alg}
\end{figure}


Furthermore, we have the following properties that will be used in the next section. Let
$n$ denote the index at which Zielonka's algorithm terminated.
The sets $W_{1-i}^j$ for $j\le n$ are obviously disjoint. However, the sets $W_{i}^{n-j}$
for $j$ in the range $n-k<j\le n$ might overlap.
Player~$i$ can confine a play in $W_i^{n-j}$
until Player~$1-i$ decides to leave this set. However, his only choice is to move
to a vertex in $\att{i}{X_{n-j}}{V \setminus \lbl(T_{n-j})}$, as he can neither move to a vertex
in $A_n=A_{n-j}$ ($X_n=X_{n-j}$ is a trap for him) nor to a vertex in $W_{1-i}^{n-j}=\emptyset$.
This implies that Player~$i$ can force the play to visit $V \setminus \lbl(T_{n-j})$,
if Player~$1-i$ decides to leave $W_1^{n-j}$.

\begin{theorem}[\cite{Z98}]
\label{thm_correctness}
Algorithm~\ref{alg_z} terminates with a partition $(W_0, W_1)$, where
Player~$0$ has a winning strategy for $W_0$ and Player~$1$ has a winning
strategy for $W_1$.
\end{theorem}

Zielonka's winning strategies are defined inductively: Player~$1-i$ plays
an attractor strategy to $U_{n-1}$ on $A_n\setminus U_{n-1}$ and on each $W_{1-i}^n$ 
according to the winning strategy computed recursively.
A play consistent with this strategy will from some point onwards be consistent 
with one of the winning strategies for some  $W_{1-i}^n$, hence it is winning for
Player~$1-i$.

Player~$i$ plays using a cyclic counter $c$: suppose $c=j$. In $W_1^{n-j}$, she plays according
to the winning strategy computed recursively. If Player~$1-i$ chooses to leave $W_1^{n-j}$,
then she starts playing an attractor strategy to reach $V\setminus\lbl(T_{n-j})$. Once
she has reached this set she increments $c$ modulo $k$ and begins again.
There are
two possibilities for a play consistent
with this strategy: if it stays from some point onwards in some $W^{n-j}_i$, then
it is winning by the inductive hypothesis. Otherwise, it will visit infinitely many
vertices in $V \setminus \lbl(\child(\ztree,j))$ for every $j$ in the range $0\le j <\branch(\ztree)$,
which implies that the infinity set of the play is not a subset of any 
$\lbl(\child(\ztree,j))$. Hence, it is in $\F_i$ and the play is indeed winning for
Player~$i$.

We conclude this section by showing that the winning strategies for 
Muller games as defined in \cite{Z98} do not bound the score of the opponent 
by a constant. 

\begin{lemma}\label{lem_unbounded}
There exists a family of Muller games $\game_n=(G_n,\F_0^n,\F_1^n)$ with 
$|G_n|=n+1$ and $|\F_0^n|=1$ such that $\maxscore_{\F_1^n}(\play(v,\sigma,\tau))=n$
where $\sigma$ is Zielonka's strategy, $v\in V$, and $\tau\in\Pi_1$.
\end{lemma}

\begin{figure}[h]
\centering
\begin{tikzpicture}[node distance=2cm,auto,every edge/.style={bend angle=20,draw,thick}]
\node[p0]			(0)	{$0$};
\node[p0, right of = 0]		(1)	{$1$};
\node[p0, right of = 1]		(2)	{$2$};
\node[right of = 2]		(dots)	{$\cdots$};
\node[p0, right of = dots]	(n)	{$n-1$};
\node[p0, right of =  n]	(n+1)	{$n$};

\path[-stealth]
(1) edge (0)
(2) edge (1)
(dots) edge (2)
(n) edge (dots)
(n+1) edge (n)
(0.north east) edge[bend left] (n+1.north)
(1.north east) edge[bend left] (n+1.north west);
\end{tikzpicture}
\caption{The arena $G_n$ for Lemma~\ref{lem_unbounded}.}
\label{pic_Gn}
\end{figure}
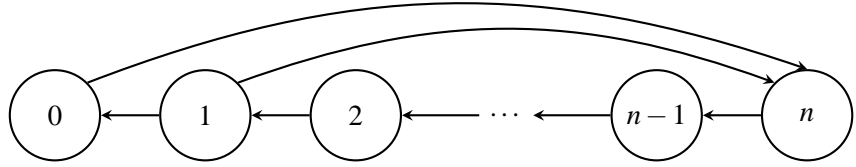

\begin{proof}
Let $G_n=(V_n,V_n,\emptyset,E_n)$ with $V_n=\{0,\ldots,n\}$, 
$E_n=\{(i+1,i)\mid i<n\}\cup\{(0,n),(1,n)\}$ (see Figure~\ref{pic_Gn}),
and $\F_0^n=\{V_n\}$.
The Zielonka tree for the winning condition $(\F_0^n,\F_1^n)$ has a root
labeled by $V_n$ and $n+1$ children that are leaves and are 
labeled by $V_n\setminus\{i\}$ for every 
$i\in V_n$. Assume, the children are ordered as follows: 
$V_n\setminus\{0\}<\cdots <V_n\setminus\{n\}$. Zielonka's strategy for 
$\game_n$, which depends on the ordering of 
the children, can be described as follows. Initialize
a counter $c:=0$ and repeat the following:
\begin{enumerate}
\item\label{start} Use an attractor strategy to move to vertex $c$.
\item Increment $c$ modulo $n+1$.
\item Go to \ref{start}.
\end{enumerate}
Now assume a play consistent with this strategy has just visited $0$.
Then, it visits all vertices
$1,\ldots,n$ in this order by cycling through the loop $n,\ldots,1$ 
$n$ times. Hence, the score for the set $\{1,\ldots,n\}$ is 
infinitely often $n$.
\end{proof}

By contrast, Player~$0$ has a positional winning strategy for $\game_n$ that bounds the
opponents scores by $2$.
The reason the strategy described above allows a high score for Player~$1$
is that it ignores the fact that, while it attracts the play to the vertex $0$,
it visits all other vertices. In the next section we will
construct a strategy that recognizes such visits. Thereby, the strategy is able to bound
the opponent's scores  by~$2$.

\section{Bounding the Scores in a Muller Game}
\label{sec_bound2}
In this section, we prove our main result: the finite-time Muller game 
with threshold $3$ is equivalent to a Muller game. 

\begin{theorem}
\label{thm_k3}
If $W_i$ is the winning region of Player~$i$ in a Muller game $(G,\F_0,\F_1)$,
and $W_i'$ is the winning region of Player~$i$ in the finite-time Muller
game $(G,\F_0,\F_1,3)$, then $W_i=W_i'$.
\end{theorem}

To prove Theorem~\ref{thm_k3} we use the following approach. If 
$\maxscore_{\F_{1-i}}(\rho) \le 2$ for an infinite play $\rho$, then there exists a prefix $w$
of $\rho$ that is winning for Player~$i$ in
the finite-time Muller game with threshold~$3$. Hence, if a winning strategy for Player~$i$
in the Muller game
bounds the scores of her opponent by $2$, then this strategy is also winning
for the finite-time Muller game with threshold $3$. 
We will show that such a winning strategy exists. Theorem~\ref{thm_k3} then follows by
determinacy of Muller games. Therefore, the rest of this section will be dedicated to
proving the following lemma.

\begin{lemma}
\label{lem_k3}
Player~$i$ has a winning strategy $\sigma$ for her winning region $W_i$ 
in a Muller game $\game=(G, \F_0,\F_1)$ such that
$\maxscore_{ \F_{1-i} }(\play(v, \sigma, \tau)) \le 2$ for every vertex~$v \in
W_i$ and every $\tau\in\Pi_{1-i}$.
\end{lemma}

We will use the internal structure of the winning regions as
computed by Zielonka's algorithm to give an
inductive proof of this claim. Traditionally, 
Zielonka's strategies forget the history of the play every time they switch
between an attractor strategy and a recursively computed winning strategy. For
example, suppose that a play $w$ spends some time in $W_1^{n-j}$
before Player~0 decides to move out of the set $W_1^{n-j}$. Player~$1$ responds
to this by playing the attractor strategy to the set $V \setminus \lbl(T_{n -
j})$ in order to reach some vertex $v \in V \setminus \lbl(T_{n - j})$. If
$v \in W_1^{n-j+1}$, then Player~$1$ will play the winning strategy for the set
$W_1^{n-j+1}$ starting at the vertex $v$.

Note that the play $w$ may have spent a significant number of steps in
$W_1^{n-j+1}$ (while playing according to the attractor strategy)
before Player~$1$ begins to play the winning strategy for that
set. Yet in Zielonka's strategy, Player~$1$ will behave as if the first vertex
visited in $W_1^{n-j+1}$ is $v$. In other words, the suffix of $w$ that is
contained in $W_1^{n-j+1}$ is effectively forgotten by the strategy. 

This
fact is irrelevant if we are only concerned with constructing a winning
strategy, but when we want to construct strategies that guarantee certain scores
are bounded by~$2$, the entire suffix of $w$ must be retained in this kind of
situation.
This motivates the following definition of a play. A play begins with a finite
prefix over which the players have no control, and then
continues as a normal play would. The key difference is that the strategies may
base their decisions on the properties of the prefix.

\begin{definition}[Play]
For a non-empty finite path $w=w_0\cdots w_n$ and strategies $\sigma\in\Pi_i$, 
$\tau\in\Pi_{1-i}$, we define the infinite play $\play(w,\sigma,\tau)=\rho_0 
\rho_1\rho_2\cdots$ inductively by $\rho_j=w_j$ for $0\le j \le n$ and for 
$j>n$ by
\[\rho_{j}=\begin{cases}
\sigma(\rho_0\cdots\rho_{j-1}) & \text{if } \rho_{j-1} \in V_i\\
\tau(\rho_0\cdots\rho_{j-1})   & \text{if } \rho_{j-1} \in V_{1-i}
\end{cases}\enspace .\]
\end{definition}

In fact, the finite paths that are passed to our strategies will not be totally
arbitrary. As described previously, these paths arise out of decisions made
before the strategy was recursively applied. Therefore, we have some control
over the form that these paths take. We will construct our strategy so that
every path passed to a recursive strategy has the following property.

\begin{definition}[Burden]
Let $\F\subseteq 
\pow{V'}$. A finite path $w$ is an $\F$-burden if $\maxscore_{\F}(w)\le 2$ and
for every $F\in\F$ either $\score_F(w)=0$ or $\score_F(w)=1$ and
$\acc_F(w)=\emptyset$.
\end{definition}

We are now ready to prove by induction over the height of the Zielonka tree that both
players have a strategy to bound their opponent's scores by $2$ on their winning regions,
even if the play starts with a burden.
We begin by considering the base
case, which is when the Zielonka tree is a leaf. For the rest of this section we
will assume $\lbl(\ztree)\in\F_1$. Otherwise, swap the roles of Player~$0$ and
$1$ below. 

\begin{lemma}\label{lem_leaf}
Let $(G,\F_0,\F_1)$ be a Muller game with vertex set $V$ such that
$\ztree$ is a leaf. Then, Player~$1$ has a strategy $\tau$ such that 
$\maxscore_{\F_{0}}(\play(wv, \sigma, \tau))
\le 2$ for every strategy $\sigma \in \Pi_{0}$ and every 
$\F_0$-burden $wv$ with $v\in V$.
\end{lemma}

\begin{proof}
As $\ztree$ is a leaf and $\lbl(\ztree)\in \F_1$ by assumption, we have 
$\F_0=\emptyset$. Hence, any strategy $\tau$ for Player~$1$ guarantees 
$\maxscore_{\F_{0}}(\play(w, \sigma, \tau)) \le 2$.
\end{proof}

We now move on to the inductive step of the proof. We will give two versions of
the inductive step, one case will be for the set $W_0$ and the other will be for
the set $W_1$. We will consider the case for the set $W_0$ first.

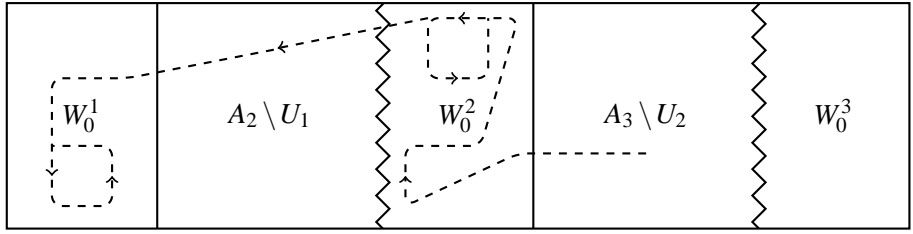
\begin{figure}[h]
\centering
\begin{tikzpicture}
\draw[thick] (0,0) rectangle (12,3);

\begin{scope}
\clip (0,0) rectangle (12,3);

\draw[thick] {(2,0) -- (2,3)};
\draw[thick] decorate [decoration =zigzag] {(5,0) -- (5,3)};
\draw[thick] {(7,0) -- (7,3)};
\draw[thick] decorate [decoration =zigzag] {(10,0) -- (10,3)};

\node at (1,1.5) {$W^1_0$};
\node at (3.5,1.5) {$A_2\setminus U_1$};
\node at (6,1.5) {$W^2_0$};
\node at (8.5,1.5) {$A_3\setminus U_2$};
\node at (11,1.5) {$W^3_0$};

\draw[thick,rounded corners, dashed,->](8.5,1) -- (6.8,1) -- (5.3,.3) -- (5.3,.7);
\draw[thick,rounded corners, dashed,->](5.3,.75) -- (5.3,1.1) -- (6.3,1.1) -- (6.8,2.8) --(6,2.8);
\draw[thick,rounded corners, dashed,->] (5.95,2.8) -- (5.6,2.8) -- (5.6,2.0) -- (6,2);
\draw[thick,rounded corners, dashed] (6.05,2) -- (6.4,2.0) --(6.4,2.8);
\draw[thick,rounded corners, dashed,->] (5.6,2.8) -- (3.6,2.4);
\draw[thick,rounded corners, dashed,->] (3.55,2.38) -- (1.6,2.0) -- (.6,2) -- (.6,.7);
\draw[thick,rounded corners, dashed,->] (.6,.65) -- (.6,.3) -- (1.4,.3) -- (1.4,.7);
\draw[thick,rounded corners, dashed] (1.4,.75) -- (1.4,1.1) -- (.6,1.1);
\end{scope}
\end{tikzpicture}
\caption{The structure of $W_0$. The dashed line indicates a play.}
\label{fig_pl0}
\end{figure}

The situation in this case is shown in Figure~\ref{fig_pl0}. Our strategy for
this case will be the same as Zielonka's strategy, but it must also deal with
the finite path that has been passed to it. We denote the attractor strategy for
Player~$0$ on $A_n\setminus U_{n-1}$ by $\sigma_n^{\text{A}}$ and we denote the
recursively computed strategy for Player~0 on $G[W^{n}_0]$ as
$\sigma_n^{\text{R}}$. We can assume that $\sigma_n^\text{R}$ satisfies the
inductive hypothesis, which means that $\maxscore_{\F_1\restriction
W^{n}_0}(\play(wv, \sigma^{\text{R}}_n, \tau)) \le 2$ for every strategy $\tau$
for Player $1$ in $G[W^{n}_0]$ and every $\F_1\restriction W^{n}_0$-burden $wv$
with $v\in W^{n}_0$. We define the following strategy $\sigma^*$ for $W_0$:
\begin{equation*}
\sigma^{*}(wv) = \begin{cases}
  \sigma_n^{\text{R}}(w'v) & \text{if $v \in W^{n}_0$ and $w'$ is the 
                           longest suffix of $w$ with $\occ(w')
                           \subseteq W^n_0$} \\
  \sigma_n^{\text{A}}(v)   & \text{if $v \in A_n \setminus U_{n-1}$} 
\end{cases}\enspace .
\end{equation*}

Our strategy chooses to use $\sigma_n^{\text{R}}$ or
$\sigma_n^{\text{A}}$ precisely when Zielonka's strategy chooses to do so.
The difference is that our strategy is careful to pass the appropriate finite
path to the recursively computed strategy $\sigma_n^{\text{R}}$.

The sets $U_j$ form a sequence of nested traps for Player~$1$. Therefore, if Player~$1$ 
chooses to leave some $U_j\setminus U_{j-1}$ and Player~$0$ plays according to $\sigma^*$, the play 
can never return to $U_j\setminus U_{j-1}$. This implies that a play that has left some $W_0^j$ 
will never return. Also, every vertex in $A_j\setminus U_{j-1}$ can be seen at most once,
as $\sigma^*$ behaves like an attractor strategy on these vertices.
The next lemma will be used to deal with cases that arise from these observations.

\begin{lemma}
\label{lem_burden}
Let $w$ be an $\{F\}$-burden,
let $v,v'\in F$.
\begin{enumerate}
\item\label{lem_burden1} Let $\rho$ be an infinite play
in which $v$ appears at most once. Then, $\maxscore_{\{F\}}(w\rho)\le 2$. 
\item\label{lem_burden2} Let $\rho$ be an infinite play such that $v$ is never 
visited after $v'$ was visited for the first time. Then, $\maxscore_{\{F\}}(w\rho)\le 2$.
\end{enumerate}
\end{lemma}
\begin{proof}
For both statements, it suffices to show that $\score_F(wx)\le 2$
for every prefix $x$ of $\rho$. Let $w=w_0\cdots w_n$. We consider
the two cases given by the definition of a burden: 
\begin{itemize} 
\item $\score_F(w)=1$. As $\acc_F(w)=\emptyset$, we have 
$\score_F(w_0\cdots w_{n-1})=0$. Hence, the suffix $w_k\cdots w_n$ of $w$ witnessing 
$\score_F(w)=1$ is minimal. 
\begin{enumerate}
\item As $w_k\cdots w_n$ is minimal and as $v$ occurs at 
most once in $\rho$, we conclude that the score for $F$ increases at 
most once after the prefix $w$.

\item As the suffix is minimal, the score of
$F$ can increase to $2$ only by or after visiting $v'$ for the
first time. But $v$ is then never visited again. Hence, the score
for $F$ is bounded by $2$.
\end{enumerate}
\item $\score_F(w)=0$. Let $y$ be the shortest prefix of 
$\rho$ such that $\score_F(wy)=1$. If such a prefix does not exist,
then we are done.
\begin{enumerate}
\item Otherwise, $y^{-1}\rho$ does contain $v$ at most once.
Hence, the score for $F$ increases at most once after the prefix $wy$.

\item Again, if such a prefix exists, then the score for
$F$ can reach $2$ only by or after visiting $v'$ for the first time after $wy$.
But $v$ is then never visited again. Hence, the score
for $F$ is bounded by $2$.\qedhere
\end{enumerate}
\end{itemize}
\end{proof}

We are now able to prove the inductive step for Player~$0$, by applying the the
observations formalized in Lemma~\ref{lem_burden} to the structure of $W_0$.

\begin{lemma}\label{lem_strip}
We have 
$\maxscore_{\F_1\restriction W_0}(\play(wv, \sigma^{*}, \tau)) \le 2$ for every strategy $\tau
\in \Pi_{1}$ and every $\F_1\restriction W_0$-burden $wv$ with $v\in W_0$.
\end{lemma}

\begin{proof}
Let $\rho=\rho_0\rho_1\rho_2\cdots= w^{-1}\play(wv,\sigma^*,\tau)$. Note that $\rho_0=v$,
which is the first vertex where the players get to choose a successor. Assume 
$\rho$ enters some $A_n\setminus U_{n-1}$. Then,
it will afterwards enter $U_{n-1}$ while seeing every vertex in 
$A_n\setminus U_{n-1}$ at most once, as $A_n$ is an attractor and
$\sigma^*$ behaves like an attractor strategy on $A_n\setminus U_{n-1}$.
Now assume $\rho$ enters some $W^n_0$. Then, it
will stay in $W^n_0$ until Player~$1$ decides to leave. However, his only 
choices are vertices in $A_{n-1}$, as $W^n_0$ is a trap for him in $V\setminus
A_{n-1}$. 
Hence, once a set $A_n\setminus U_{n-1}$ or $W^n_0$ is left, it will never be
entered again.

As $w\rho_0$ is an $\F_1\restriction W_1$-burden, it suffices to show 
$\score_F(w\rho_0\cdots \rho_n)\le 2$ for every $n > 0$ and every $F\in 
\F_1\restriction W_1$. We will consider several cases for $F$: remember that either 
$\score_F(w\rho_0)=0$ or $\score_F(w\rho_0)=1$ and $\acc_F(w\rho_0)=\emptyset$.
\begin{itemize}
\item $F\cap\left(\bigcup_{n\ge 1}\left(A_n\setminus U_{n-1}\right)
\right)\not=\emptyset$: Every vertex in $\bigcup_{n\ge 1}\left(A_n\setminus
U_{n-1}\right)$ occurs at most once in $\rho$. Hence, 
$\score_F(w\rho_0\cdots \rho_n)\le 2$ for every $n > 0$ by Lemma~\ref{lem_burden}.\ref{lem_burden1}.

\item $F\subseteq \bigcup_{n\ge 1}W^n_0$ with $F\cap W_0^i\not=
\emptyset$ and $F\cap W_0^j\not=\emptyset$ for $i<j$: $\rho$
cannot visit $W_j$ after it has visited $W_i$. Thus, $\score_F(w\rho_0\cdots \rho_n)
\le 2$ for every $n > 0$ by Lemma~\ref{lem_burden}.\ref{lem_burden2}.

\item $F\subseteq W^j_0$ for some $j$: If $\rho$ never visits $W^j_0$, 
then $\score_F(w\rho_0\cdots \rho_n)=0$ for every $n > 0$. So, assume 
$\rho$ enters $W^j_0$ at position $\rho_m$ for some $m \ge 0$.

Suppose  $m=0$: $w\rho_0$ is also an $\F_1\restriction W_0^j$-burden and
$w\rho_0\rho_1\rho_2\cdots$ is played according to $\sigma_j^{\text{R}}$ 
until Player~$1$
decides to leave $W^j_0$ at some position $p>m$.
Applying the inductive hypothesis
yields that
$\sigma_j^{\text{R}}$ guarantees
$\score_F(w\rho_0\cdots \rho_n)\le 2$ for every $n$ in the range $m \le n \le p$.
Should the play leave 
$W^j_0$, then $\score_F$ is reset to $0$ and stays $0$, as $W^j_0$ 
cannot be visited again. If Player~$1$ never leaves $W^j_0$, then the
scores are bounded by $2$ throughout the whole play.

If $m>0$, then $\score_F(w\rho_0\cdots \rho_n)=0$ for every 
$n < m$. Also, the play $\rho_m\rho_{m+1}\rho_{m+2}\cdots$
in $W^j_0$ starts with the $\F_1\restriction W_0^j$-burden $w\rho_0\cdots\rho_m$, 
(as $\rho_{m-1}\notin W_0^j$)
and the inductive hypothesis on $\sigma_j^{\text{R}}$
guarantees $\score_F(w\rho_0\cdots \rho_n) \le 2$ until $W^j_0$ is left, from which
point onwards $\score_F$ is always $0$.\qedhere
\end{itemize}
\end{proof}

We now turn our attention to the strategy for Player~$1$. For the rest of this
section $n$ will be the index at which Zielonka's algorithm terminated, and
$k= \branch(\mathcal{Z}_{\F_0, \F_1})$. The situation for Player~$1$ consists of
$k$ overlapping instances, one for each child, of the situation depicted
in Figure~\ref{fig_pl1}.

\begin{figure}[h]
\centering
\begin{tikzpicture}

\draw[thick] (0,0) rectangle (12,3);
\begin{scope}
\clip (0,0) rectangle (12,3);
\draw[thick] {(7,0) -- (7,3)};
\draw[thick] decorate [decoration =zigzag] {(0,1.8) -- (7,1.8)};
\node at (3.5,.9) {$W_1^{n-j}$};
\node at (3.5,2.5) {$\att{1}{X_{n-j}}{V\setminus \lbl(T_{n-j})}$};
\draw[thick,(-)] (7,-.3) -- (12,-.3);
\node at (9.5,1.5) {$(V\setminus \lbl(T_{n-j}))\cap X_{n-j}$};

\draw[thick,rounded corners, dashed,->] (6,1.5) -- (5,1) -- (4,1) -- (4,.6);
\draw[thick,rounded corners, dashed,->] (4,.55) -- (4,.2) -- (5,.2) -- (5,.6);
\draw[thick,rounded corners, dashed]    (5,.65) -- (5,1);
\draw[thick,rounded corners, dashed,->] (4,.2) -- (2,.4) -- (2,.9);
\draw[thick,rounded corners, dashed,->] (2,.95) -- (2,1.6) -- (2.7,1.6) -- (2.7,.9);
\draw[thick,rounded corners, dashed,->] (2.7,.85) -- (2.7,.4) -- (2,.4) -- (1.0,.9) -- (1.0,1.5);
\draw[thick,rounded corners, dashed,->](1.0,1.55) -- (1.0,2.1) -- (7.3,2.1);

\end{scope}
\end{tikzpicture}
\caption{The structure of $W_1$ with respect to $T_{n-j}$. The dashed line indicates a part of a play between two change points.}
\label{fig_pl1}
\end{figure}
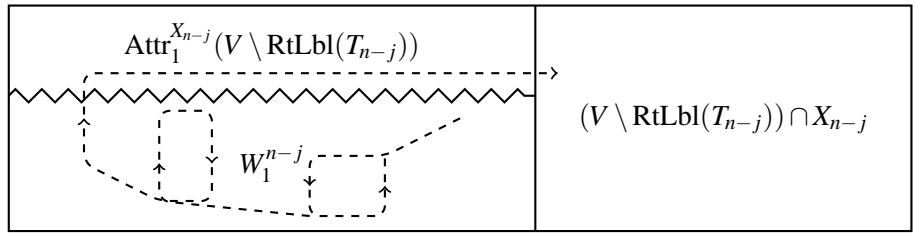

For the sake of convenience we define $Z_j = X_j \setminus \lbl(T_{j})$ for
every $j$ in the range $n-k< j \le n$. For every $j$ in the range $n - k < j \le
n$, we have an attractor strategy for Player~$1$ on $\att{1}{X_j}{Z_j}\setminus
Z_j$ which we call $\tau_j^\text{A}$, and we have a recursively computed winning
strategy $\tau_j^\text{R}$ for Player~$1$ on $G[W^{j}_1]$. Once again,
we can assume the inductive hypothesis holds for the strategy $\tau_j^\text{R}$,
which means that $\maxscore_{\F_0 \restriction W_1^j}(\play(wv, \sigma,
\tau_j^{\text{R}}) \le 2$ for every strategy $\sigma$ of Player~$0$ in
$G[W^{j}_1]$ and every $\F_0\restriction W_1^j$-burden $wv$ with  $v\in W_1^j$.

Our strategy improves the strategy given by Zielonka in the sense that it uses
a different method for choosing a new child of the root. Zielonka's
strategy works through the children in a cyclic order, which means that when the
play enters the set $V \setminus \lbl(T_j)$ the strategy will then move on to
the child $T_{j+1}$, and begin playing either $\tau_{j+1}^{\text{A}}$ or
$\tau_{j+1}^{\text{R}}$. By contrast, we will use a more careful method
for picking the next child of the root that will be considered.

Our method for picking the next child will make its decision based on which sets
of the opponent have either non-zero score or a non-empty accumulator.
For this purpose, we define the indicator function of a play $\ind\colon
V^+\rightarrow \pow{V}$ as
\begin{equation*}
\ind(w)=\bigcup_{\substack{F\in \F_0\colon\\ \score_F(w)>0}}F\,\cup
\bigcup_{\substack{F\in \F_0\colon\\ \acc_F(w)\not=\emptyset}}\acc_F(w)\enspace.
\end{equation*}
%

Recall that Lemma~\ref{lem_chain} implies that the sets we are considering form
a chain in the subset relation. This implies that the indicator function always
gives some subset of a set that belongs to the opponent. Therefore, we can argue that there
must always exist a child of the root whose label contains the indicator set.

\begin{lemma} 
For every $w$, there is some~$j$ in the range $n-k <  j \le n$ such that $\ind(w) 
\subseteq \lbl(T_j)$.
\end{lemma}
\begin{proof}
Lemma~\ref{lem_chain} implies that there is a maximal set $G$ such 
that $\ind(w)=G$, with either $\score_G(w)>0$ or $\acc_F(w)=G$ for some 
$F\in\F_0$ with $G\subseteq F$. Hence, $\ind(w)\subseteq F$ for some 
$F\in\F_0$, and, by definition of $\ztree$, there is some child of the 
root labeled by $\lbl(T_j)$ such that $F\subseteq \lbl(T_j)$.
\end{proof}

When a new child must be chosen, our strategy will choose some child whose label
contains the value of the indicator function for the play up to that point. It is also
critically important that this condition is used when picking the child in the
first step. This is the part of the strategy where the finite initial path can
have an effect on the decisions that the strategy makes.

We can now formally define this strategy. We begin by defining an auxiliary function
that specifies which child the strategy is currently considering. We define
$c:W_1^*\rightarrow\{n - k + 1, \dots, n,\bot\}$ as $c(\epsilon)=\bot$ and
\begin{equation*}
c(wv)=\begin{cases} 
c(w) & \text{if }v \in \lbl(T_{c(w)})\\
j    & \text{if }v \notin \lbl(T_{c(w)})\text{, }\ind(wv) \ne \emptyset
       \text{ and } j \text{ minimal with }\ind(wv) \subseteq \lbl(T_{j}) \\
j    & \text{if }v \notin \lbl(T_{c(w)})\text{, }\ind(wv)=\emptyset 
       \text{ and } j \text{ minimal with } v\in \lbl(T_{j})\\
\bot & \text{if }v\not\in\bigcup_{n - k < j\le n}\lbl(T_{j})
\end{cases}\enspace.
\end{equation*}
Now we can define $\tau^{*}$ for $W_{1}$ as
\begin{equation*}
\tau(wv)^{*} =\begin{cases}
\tau^\text{R}_j(wv) & \text{if }c(wv)=j, v\in W^{j}_1 \text{ and } w' 
                      \text{ is the longest suffix of $w$ with } \occ(w')
                             \subseteq W^{j}_1\\
\tau^\text{A}_j(v)  & \text{if }c(wv)=j, v\in \lbl(T_j) \setminus W^j_1 \\
x                   & \text{if }c(wv)=\bot \text{ where $x\in W_1$ with $(v,x)\in E$}
\end{cases}\enspace.
\end{equation*}

We will now prove that this strategy has the required properties. Our proof will
use the concept of a change point, which is a position in a play where the $c$
function changes. More formally, suppose that
$\rho=\rho_0\rho_1\rho_2\cdots = w^{-1}\play(wv,\sigma,\tau^*)$ for some $\F_0$-burden $wv$ with $v\in W_1$ and
$\sigma\in\Pi_0$. Note that
$\rho_0=v$, which is the first vertex at which the players get to choose the
successor. We say that a position $r$ of $\rho$ is a change point, if $r=0$ or
if $c(w\rho_0\cdots\rho_{r-1})\not=c(w\rho_0\cdots\rho_r)$.

Let $x$ be a finite prefix of an infinite play that is consistent with $\tau$
such that the last position in $x$ is a change point. Moreover, assume that $x$
satisfies the burden property. Our strategy will pick some index~$j$ such that
$\ind(x) \subseteq \lbl(T_j)$. The play will then remain in the set $W_1^{j}$ until
Player~$0$ chooses to leave the set $W_1^{j}$, at which point the strategy
attracts to the set $V \setminus \lbl(T_j)$. Once such a vertex is reached,
the scores for all sets $F\in\F_0$ with $\score_F(x)>0$ are reset to $0$ and
the accumulator for $F$ is empty for every $F\in\F_0$ with $\acc_F(x)\not=\emptyset$.
While attracting the play to $V \setminus \lbl(T_j)$ the scores for other sets $F\in\F_0$
might rise and the accumulators fill up. However, as every vertex in the attractor is seen at
most one, we are able to show the following: if the play up to a change point is a
$\F_0$-burden, then the play up to the next change point is also a burden. As a $\F_0$-burden
bounds the scores of Player~$0$ be $2$, this suffices to prove that $\tau^*$ bounds Player $0$'s
scores by $2$.


\begin{lemma}\label{lem_cp}
Let $\rho$ be as above and let $r<s$ be two change-points such that there 
exists no change point $t$ with $r<t<s$. If $w\rho_0 \cdots \rho_{r}$ is 
an $\F_0\restriction W_1$-burden, then so is $w\rho_0 \cdots \rho_{s}$. 
\end{lemma}

\begin{proof} 
From the definition of a change point we get 
$c(\rho_0\cdots\rho_t)=c(\rho_0\cdots\rho_r)$ for every $t$ in the range $r<t<s$.

If $c(w\rho_0\cdots\rho_r)=\bot$, then 
$\rho_t\not\in\bigcup_{n - k < j\le n}\lbl(T_{j})$, which implies $\rho_t\notin F$ for 
every $F\in\F_0$. Hence, we have $\score_F(w\rho_0\cdots\rho_{t})=0$ for 
every $r\le t< s$ and every $F\in\F_0$. Furthermore, we have either
$\score_F(w\rho_0\cdots\rho_{s})=0$, if $F\not=\{\rho_s\}$ and
$\score_F(w\rho_0\cdots\rho_{s})=1$ and $\acc_F(w\rho_0\cdots\rho_{s})=\emptyset$
otherwise.

Now, assume $c(w\rho_0\cdots\rho_r)=j$ for some $j$ in the range $n-k<j\le n$. Then,
there exists an $u$ in the range $r\le u\le s$ such that 
$\rho_r\cdots\rho_{u-1}$ is in $W_1^j$, $\rho_{u}\cdots\rho_{s-1}$ 
is in $\att{1}{X_j}{Z_j}\setminus Z_j$, and we have $\rho_s\notin \lbl(T_j)$.
Note that both parts could be empty.
The situation is depicted in Figure~\ref{fig_decomp} (cf. also Figure~\ref{fig_pl1}).

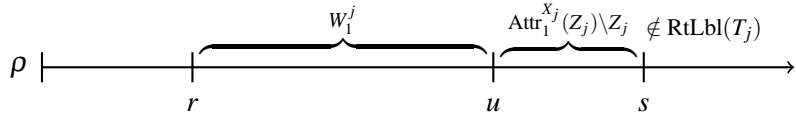
\begin{figure}[h]
\centering
\begin{tikzpicture}
\draw[->, thick] (0,0) -- (10,0);
\draw[thick] (2,-.2) -- (2,.2);
\draw[thick] (8,-.2) -- (8,.2);
\draw[thick] (6,-.2) -- (6,.2);
\draw[thick] (0,-.2) -- (0,.2);

\draw node at (-.3,0) {$\rho$};

\draw node at (2,-.5) {$r$};
\draw node at (8,-.5) {$s$};
\draw node at (6,-.5) {$u$};

\draw node at (4,.5) {$\overbrace{\hspace{3.8cm}}^{W_1^j}$};
\draw node at (7,.5) {$\overbrace{\hspace{1.8cm}}^{\att{1}{X_j}{Z_j}\setminus Z_j}$};
\draw node at (8.8,.5) {\footnotesize  $\notin \lbl(T_j)$};
\end{tikzpicture}
\caption{A part of a play between two change points.}
\label{fig_decomp}
\end{figure}

Furthermore, at positions $i$ in the range $r\le i \le u-1$, Player~$1$
plays according to $\tau^{\text{R}}_j$ and positions $i$ in the range 
$u\le i \le s-1$, he plays according to $\tau^{\text{A}}_j$. This implies
that every vertex in $\att{1}{X_j}{Z_j}\setminus Z_j$ is seen at most once 
in between $\rho_r$ and $\rho_s$,
i.e., in the infix $\rho_{u}\cdots\rho_{s-1}$.

Finally, let $\ind(w\rho_0\cdots\rho_r)=G$. If $G\not=\emptyset$, then 
$G\subseteq \lbl(T_j)$; otherwise, $\rho_r\in \lbl(T_j)$, both by definition
of $c$.

It suffices to show for every $F\in\F_0\restriction W_1$:
\begin{enumerate}
\item\label{ad1} If $\score_F(w\rho_0\cdots\rho_{r})=1$ and $\acc_F(w\rho_0\cdots\rho_{r})=
\emptyset$, then $\score_F(w\rho_0\cdots\rho_{t})\le 2$ for all $r<t<s$ and 
$\score_F(w\rho_0\cdots\rho_{s})=0$.

\item\label{ad2} If $\score_F(w\rho_0\cdots\rho_{r})=0$, then $\score_F(w\rho_0\cdots
\rho_{t})\le 2$ for all $r<t<s$ and either $\score_F(w\rho_0\cdots\rho_{s})=1$
 and $\acc_F(w\rho_0\cdots\rho_{s})=\emptyset$ or $\score_F(w\rho_0\cdots 
\rho_{s})=0$.
\end{enumerate}
\ref{ad1}: As $\emptyset\not=F\subseteq 
\ind(w\rho_0\cdots\rho_r)$, we have $F\subseteq \lbl(T_j)$ and hence
$\score_F(w\rho_0\cdots\rho_s)=0$, as $\rho_s\in Z_j=X_j \setminus \lbl(T_{j})$.
It remains to show $\score_F(w\rho_0\cdots\rho_{t})\le 2$ for all $r<t<s$. 
We consider several cases for $F$:
\begin{itemize}
\item $F\cap Z_j\not=\emptyset$: as the vertices in $Z_j$ are not visited by
$\rho_{r}\cdots\rho_{s-1}$, the score of $F$ cannot increase in this interval.
\item $F\cap \att{1}{X_j}{Z_j}\setminus Z_j\not=\emptyset$: every vertex
in the attractor is seen at most once. Hence, we obtain 
$\score_F(w\rho_0\cdots\rho_{t})\le 2$ for all $r<t<s$ by Lemma~\ref{lem_burden}.\ref{lem_burden1}.
\item $F\subseteq W^j_1$: If $\rho_r\in\att{1}{X_j}{Z_j}\setminus Z_j$, then
$\score_F$ is reset to $0$ at $\rho_r$ and stays $0$ until $\rho_s$, hence, 
we have $\score_F(w\rho_0\cdots\rho_{t})\le 2$ for all $r<t<s$.

So, suppose $\rho_r\in W^j_1$. As $w\rho_0\cdots\rho_r$ is also an $\F_0\restriction W_1^j$-
burden and as 
$F\in \F_0\restriction W_1^j$, the inductive hypothesis on $\tau_1^j$ guarantees 
$\score_F(w\rho_0\cdots\rho_{t})\le 2$ for every $r<t<u$. As
$\rho_t\notin W^j_1$ for every $t$ in the range $u\le t < s$, we also have 
$\score_F(w\rho_0\cdots\rho_{t}) = 0$ for these positions.
\end{itemize}
\ref{ad2}: Let $G=\acc_F(w\rho_0\cdots\rho_r)\subseteq \ind(w\rho_0\cdots\rho_r)$.
Note that $G\subseteq \lbl(T_j)$, but it could be the case that $F\not\subseteq \lbl(T_j)$. 
Again, we consider several cases for $F$:
\begin{itemize}
\item If $\rho_s\in F$, then $\rho_s\notin \acc_F(w\rho_0\cdots\rho_r)$, as 
$\acc_F(w\rho_0\cdots\rho_r)\subseteq \lbl(T_j)$ and $\rho_s\notin \lbl(T_j)$. Hence,
 $\score_F$ stays $0$ at every position between $r$ and (excluding) $s$, as the vertex $\rho_s$
is never visited. If $\score_F(w\rho_0\cdots\rho_s)=1$, then $\acc_F(w\rho_0\cdots\rho_s)=
\emptyset$; otherwise $\score_F$ is $0$ at position $s$, too.  
\item If $\rho_s\notin F$, then $\score_F(w\rho_0\cdots\rho_{s}) = 0$. To bound the score between the positions $r$ and $s$ by $2$, we have to consider three
subcases: either $F\cap Z_j\not=\emptyset$, $F\cap \att{1}{X_j}{Z_j} \setminus Z_j\not=\emptyset$ 
or $F\subseteq W^j_1$. 
All cases can be solved by analogous reasoning to these cases in \ref{ad1}.\qedhere
\end{itemize}
\end{proof}

Now, to prove the inductive step for $W_1$, we simply need to observe that the
finite path ending at the first change point is a burden by assumption.

\begin{lemma}\label{lem_alt}
We have 
$\maxscore_{\F_0\restriction W_1}(\play(wv, \sigma, \tau^{*})) \le 2$ for every strategy $\sigma
\in \Pi_{0}$ and every $\F_0\restriction W_1$-burden $wv$ with $v\in W_1$.
\end{lemma}

\begin{proof}
Let $\rho=w^{-1}\play(wv, \sigma, \tau^{*})$. If $\rho$ contains
infinitely many change points, then Lemma~\ref{lem_cp} implies 
$\maxscore_{\F_0}(\play(wv, \sigma, \tau^{*})) \le 2$ as the play 
starts with a burden, i.e., there scores are bounded by $2$ in $wv$,
and in between any two
change points, the scores are bounded by $2$ as well. 
If $\rho$ contains only finitely many change points, then Lemma~\ref{lem_cp}
implies that the scores of Player~$0$ up to the last change point are
bounded by $2$. From that point onwards, $\play(wv, \sigma, \tau^{*})$ is
consistent with some $\tau_1^j$, and the play up to that point is a 
$\F_1\restriction W_1^j$-burden, as it is an $\F_1\restriction W_1$-burden
 due to Lemma~\ref{lem_cp}. 
Hence, the scores for every set $F\in \F_1\restriction W_1^j$
are bounded by $2$ from that point onwards, by the inductive hypothesis on $\tau_1^j$.
The scores of every $F\in\F_1\restriction W_1$ with $F\not\subseteq W_1^j$ are bounded
by $1$, as vertices not in $W_1^j$ are no longer visited.
\end{proof}
Finally, we can prove Lemma \ref{lem_k3}, which also completes the proof of
Theorem~\ref{thm_k3}.

\begin{proof} 
Theorem~\ref{thm_correctness} yields that algorithm~\ref{alg_z} is 
correct, i.e. the sets $W_i$ returned are indeed the winning regions of the 
players. We prove the following stronger statement by induction over the height
of $\ztree$: let $V$ be the vertex set of $G$.
\textit{Player~$i$ has a winning strategy $\sigma$ for her winning region $W_i$
in $G$ such that $\maxscore_{\F_{1-i}\restriction V}(wv,\sigma,\tau)\le 2$
 for every strategy $\tau\in\Pi_{1-i}$
 and every $\F_{1-i}\restriction V$-burden $wv$ in $G$.}
This implies Lemma~\ref{lem_k3}, as the finite play $v$ for every $v\in W_i$ 
is an $\F_{1-i}\restriction V$-burden.

For the induction start, apply Lemma~\ref{lem_leaf}. In the induction step, 
use the strategies obtained from the induction hypothesis to define $\sigma^*$ 
and $\tau^*$ as above and apply Lemma~\ref{lem_strip} respectively 
Lemma~\ref{lem_alt}. Both strategies are winning, as they bound the scores of the opponent by $2$.
\end{proof}

\section{Conclusion}
\label{sec_conc}

We have presented a criterion to stop plays in a Muller game after a
finite amount of time that preserves winning regions. Our bound $3^{|G|}$
on the length of a play improves the bound $|G|!+1$ obtained by a reduction to parity
games. Furthermore, our techniques show that the winning player can bound the scores
of the opponent by $2$ and that this bound $2$ is tight.

However, it remains open whether a play can also be stopped
after a score of $2$ is reached. As the winning player cannot
always avoid a score of $2$ for the opponent, one has to
show that the winning player always reaches a score of $2$ for
one of her sets before the opponent reaches score $2$ for one
of his sets. Our approach does not seem to be suitable for this, as 
the notion of a burden is not sufficient for this goal. Furthermore,
it is unclear how to strengthen the definition while still retaining Lemmata
corresponding to Lemma~\ref{lem_strip} and Lemma \ref{lem_alt}. 

A finite-time Muller game with threshold $k$ is a reachability game (in the unraveling
of the original arena up to depth at most $k^{|G|}$), which can be solved 
with simple algorithms. Another interesting direction for research is to
find a construction which turns a winning strategy for a finite-time Muller game
with threshold $3$ (or $2$, if it is equivalent) into a finite-state strategy 
for the original Muller game. It is conceivable that such a construction would yield
memory structures that are optimized for a given arena, something which does not hold
for the LAR respectively Zielonka tree structures.

\paragraph{Acknowledgements.}
The authors want to thank Wolfgang Thomas for bringing McNaughton's work
to their attention and Marcus Gelderie, Michael Holtmann, Marcin Jurdzi\'nski,
and J\"org Olschewski for fruitful discussions on the topic. Also, they want to
thank the referees for their helpful comments.


\bibliographystyle{eptcs}
\bibliography{biblio}{}

\end{document}